\documentclass[a4paper,USenglish,cleveref, autoref, thm-restate]{lipics-v2021} 
\usepackage{tcolorbox}
\usepackage{tikz}
\usepackage{mathtools}
\usepackage{caption}
\usepackage{graphicx}

\graphicspath{{Figures/}}

\newcommand{\Occurrences}{\mathsf{Occ}}

\newcommand{\LCS}{\mathsf{LCS}}
\newcommand{\ALDS}{\mathsf{ALDS}}
\newcommand{\LDS}{\mathsf{LDS}}
\newcommand{\LIS}{\mathsf{LIS}}

\newcommand{\Otild}{\tilde{O}}

\newcommand{\Exclude}{\mathsf{Exclude}}
\newcommand{\Project}{\mathsf{Project}}

\newcommand{\Init}{\mathsf{Init}}

\newcommand{\Delete}{\mathsf{Delete}}

\newcommand{\RF}{\mathsf{RF}}

\newcommand{\BestMatch}{\mathsf{BestMatch}}
\newcommand{\Match}{\mathsf{Match}}

\newcommand{\polylog}{{\mathsf{polylog}}}

\newcommand{\eps}{\varepsilon}

\DeclarePairedDelimiter\floor{\lfloor}{\rfloor}

\newcommand{\para}[1]{\subparagraph*{#1}}

\newcounter{problemctr}
\renewcommand{\theproblemctr}{\arabic{problemctr}}

\crefname{algocfline}{line}{lines}
\Crefname{algocfline}{Line}{Lines}

\crefname{problemctr}{Problem}{Problems}

\crefname{claim}{Claim}{Claims}

\crefname{observation}{Observation}{Observations}

\title{Deterministic Longest Common Subsequence Approximation in Near-Linear Time}
\author{Itai Boneh}{Reichman University and University of Haifa, Israel \and \url{https://sites.google.com/view/itai-boneh/home-page}}{itai.bone@biu.ac.il}{https://orcid.org/0009-0007-8895-4069}{supported by Israel Science Foundation grant 810/21.}

\author{Shay Golan}{Reichman University and University of Haifa, Israel \and \url{https://sites.google.com/view/shaygolan}}{shayg@ariel.ac.il}{https://orcid.org/0000-0001-8357-2802}{supported by Israel Science Foundation grant 810/21.}

\author{Matan Kraus}{Bar Ilan Univesity, Israel}{matan3@gmail.com}{https://orcid.org/0000-0002-2989-1113}{supported by the ISF grant no. 1926/19, by the BSF grant 2018364, and by the ERC grant MPM under the EU's Horizon 2020 Research and Innovation Programme (grant no. 683064).}
 
\authorrunning{Boneh, Golan and Kraus}


\Copyright{Jane Open Access and Joan R. Public}


\category{} 

\relatedversion{}
\nolinenumbers 

\EventEditors{John Q. Open and Joan R. Access}
\EventNoEds{2}
\EventLongTitle{42nd Conference on Very Important Topics (CVIT 2016)}
\EventShortTitle{CVIT 2016}
\EventAcronym{CVIT}
\EventYear{2016}
\EventDate{December 24--27, 2016}
\EventLocation{Little Whinging, United Kingdom}
\EventLogo{}
\SeriesVolume{42}
\ArticleNo{23}
\hideLIPIcs

\begin{document}

\maketitle
\begin{abstract}
We provide a deterministic algorithm that outputs an $O(n^{3/4} \log n)$-approximation for the Longest Common Subsequence (LCS) of two input sequences of length $n$ in near-linear time.
This is the first \emph{deterministic} approximation algorithm for LCS that achieves a sub-linear approximation ratio in near-linear time.
\end{abstract}

\newpage

\section{Introduction}



The Longest Common Subsequence ($\LCS$) problem is a classic string comparison task, whose standard dynamic programming solution is a staple of introductory computer science curricula.
The $\LCS$ problem and its variations have been studied extensively over the past decades, including in works such as  \cite{WF74, AhoHU76, Hirschberg77, HuntS77, MasekP80, NakatsuKY82, Apostolico86, Myers86, ApostolicoG87, EppsteinGGI92, BergrothHR00, IliopoulosR09, ABWV15, BringmannK15, AbboudHWW16, BringmannK18, AbboudB18, RubinsteinSSS19, RS20, HajiaghayiSSS19}.
Given two sequences, the goal is to find a longest sequence that appears as a subsequence in both, a task fundamental to understanding similarity and structure in discrete data. 
Classical algorithms solve $\LCS$ in quadratic time, specifically $O(n^2)$ for sequences of length $n$, using dynamic programming \cite{WF74,Hir75}. 
However, despite decades of research, improving this bound to truly subquadratic time remains elusive, with conditional lower bounds based on the Strong Exponential Time Hypothesis (SETH) suggesting that no truly subquadratic algorithm exists for general instances \cite{ABWV15}.
These hardness results underscore the intrinsic complexity of $\LCS$ and motivate the exploration of approximation algorithms and specialized cases to circumvent these barriers.

In this work, we consider the problem of approximating $\LCS(x, y)$ of two input sequences $x$ and $y$ in near-linear time.
Specifically, the goal is to output a common subsequence of $x$ and $y$ whose length approximates the length of a longest common subsequence.
For sequences over a binary alphabet, it is trivial to obtain a $1/2$-approximation for $\LCS(x, y)$. 
This trivial $1/2$-approximation ratio remained the best known in this setting until recently, when Rubinstein and Song~\cite{RS20} showed that a slightly better approximation ratio can be achieved for binary sequences of the same length.
The approximation algorithm of Rubinstein and Song is essentially a reduction to Edit Distance approximation.
When plugging in the state-of-the-art Edit Distance approximation algorithm by Andoni and  Nosatzki~\cite{AN20}, their algorithm achieves a $1/2+\eps$ approximation ratio in near-linear time.

Subsequently, Akmal and Vassilevska Williams~\cite{AVW21} presented a subquadratic-time algorithm achieving a better-than-$1/2$ approximation ratio even when the input sequences are allowed to have different lengths. They further generalized their result to obtain a subquadratic algorithm that achieves an approximation ratio better than $1/k$ for sequences over an alphabet of size $k$.
However, for general alphabets, the algorithms of~\cite{RS20,AVW21} do not yield an approximation ratio polynomially better than $1/n$, a trivial guarantee obtainable by selecting a single common character.

In the fully general setting, where no assumptions are made about the alphabet size, a folklore algorithm achieves an $O(\sqrt{n})$-approximation\footnote{There is some discrepancy in the literature regarding how to denote the approximation ratio of an algorithm. In works on small alphabets, the term `$\alpha$-approximation' typically refers to an algorithm that returns a common subsequence of length at least $\alpha L$, where $L$ is the length of the longest common subsequence. In contrast, works on general input sequences often refer to an `$X$-approximation' as an algorithm that returns a subsequence of length at least $L/X$. We adopt the latter convention for the rest of this paper.} in $\Otild(n)$ time\footnote{Throughout this paper, $\Otild(f(n)) = O(f(n) \cdot \polylog n)$.}.
More recently, Hajiaghayi, Seddighin, and Seddighin and Sun~\cite{HSSS22} presented a linear-time algorithm that achieves a slightly better approximation ratio of $O(n^{0.497956})$. Shortly thereafter, Bringmann, Cohen{-}Addad and Das~\cite{BCD23} proposed an improved algorithm, offering an $O(n^{0.4})$-approximation with similar running time.

Interestingly, all known near-linear time algorithms for approximating $\LCS$ in the general setting, where the alphabet may be arbitrarily large, are inherently randomized. This includes the folklore $O(\sqrt{n})$-approximation. Despite the fundamental nature of the $\LCS$ problem and the significant attention its approximation has received in recent years, no \emph{deterministic} near-linear time approximation algorithm is known.
We present an algorithm with $\Otild(n)$ running time, where $n = |x| + |y|$ is the total length of the input sequences, that outputs an $O(n^{3/4} \log n)$-approximation of $\LCS(x, y)$.
This is the first $\LCS$ approximation algorithm to achieve a non-trivial approximation ratio in near-linear time without relying on randomness.
Our result is summarized in the following theorem.
\begin{theorem}\label{thm:main}
    There is a deterministic algorithm that receives two input sequences $x$ and $y$ with $n=|x|+|y|$, and returns in $\Otild(n)$ time a common subsequence $L$ of $x$ and $y$ such that $|L| \ge |\LCS(x,y)|/(n^{3/4}\log n)$.
\end{theorem}
In \cref{sec:prelim}, we present useful notation and pre-existing tools.
In \cref{sec:Warm-up}, we present a simplified version of our main algorithm achieving an $O(n^{4/5})$-approximation.
This version illustrates our novel technique, Greedy LDS peeling, that allows us to approximate $\LCS$ deterministically.
In \cref{sec:main} we enhance the simplified algorithm to finally obtain the $O(n^{3/4}\log n)$-approximation and prove \cref{thm:main}. 


\section{Preliminaries}\label{sec:prelim}

For a natural number $n\in\mathbb N$, we denote $[n]=\{1,2,\dots,n\}$.
We also denote consecutive ranges of integers as $\{ a,a+1,\ldots b\}= [a..b]$.
Throughout this paper, we denote the set of symbols in the input sequences as $\Sigma$.
We denote a sequence $x$ over an alphabet $\Sigma$ as $x=x_1,x_2,\dots,x_{|x|}$.
Given a sequence $x$ we say that $\hat{x}=x_{i_1},x_{i_2},\dots,x_{i_k}$ is a \emph{subsequence} of $x$ if $i_1<i_2<\dots<i_k$.
Given two sequences $x$ and $y$, we say that $z$ is a common subsequence of $x$ and $y$ if $z$ is a subsequence of $x$ and also a subsequence of $y$.
We say that $z$ is a longest common subsequence ($\LCS$) of $x$ and $y$ if for every $\hat{z}$ that is a common subsequence of $x$ and $y$, it holds that $|z|\ge |\hat{z}|$.
We denote $\LCS(x,y)$ as some longest common subsequence of $x$ and $y$.
Note that while the notation $\LCS(x,y)$ is ambiguous, as there are possibly several different longest common subsequences, the term $|\LCS(x,y)|$ is not - as all longest common subsequences of $x$ and $y$ have the same length.

For a symbol $\sigma \in \Sigma$ and a sequence $x \in \Sigma^*$, we denote by $\#_{\sigma}(x) = |\{i\in [|x|]\mid x_i=\sigma \}|$  the number of occurrences of the symbol $\sigma$ in $x$.
We define two functions, both of which were presented by Rubinstein and Song~\cite{RS20}.
Moreover,  both functions can be simply computed in near-linear time.

\begin{definition}[$\Match(x,y,\sigma)$] Let $x$ and $y$ be two sequences and let $\sigma\in\Sigma$ be a symbol, $\Match(x,y,\sigma)=\min\{\#_\sigma(x),\#_\sigma(y)\}$.
\end{definition}

\begin{definition}[$\BestMatch(x,y)$]\label{lem:BestMatch}
Let $x$ and $y$ be two sequences over $\Sigma$, and let $\bar\sigma=\arg\max_{\sigma\in\Sigma}\{\Match(x,y,\sigma)\}$.
Then, $\BestMatch(x,y)= \bar\sigma^{{\Match(x,y,\bar\sigma)}}$.
\end{definition}


\para{Longest Increasing Subsequence}
Let $x=x_1,x_2,\ldots ,x_n$ be a sequence of elements and let $\prec$ be a total order over the elements of $x$.
A $\prec$-increasing subsequence of $x$ is a subsequence of increasing elements in $S$ with respect to $\prec$.
Formally, an increasing subsequence of $S$ specified as an increasing sequence $i_1<i_2< \ldots< i_{\ell}$ of indices such that $x_{i_1} \prec x_{i_2} \prec \ldots \prec x_{i_{\ell}}$. 
We call $\ell$ the length of the increasing subsequence.
A Longest Increasing Subsequence (LIS) of $x$ with respect to $\prec$, denoted by $\LIS_{\prec}(x)$, is a $\prec$-increasing subsequence of $x$ of maximum length among all $\prec$-increasing subsequences of $x$.
One can symmetrically define the Longest Decreasing Subsequence of $x$ with respect to $\prec$, denoted as $\LDS_\prec(x)$, as a maximal length subsequence such that every element is smaller (according to $\prec$) than its predecessor in the sequence.

It is well known that $\LIS_{\prec}(x)$ can be computed for an input sequence $x$ of length $n$ in $\Otild(n)$ time \cite{Knuth73,F75}. 
Our algorithm requires fast approximated LIS computation in the (partially) dynamic settings, where elements are deleted from $x$ and we wish to approximate $\LIS_\prec(x)$ or $\LDS_\prec(x)$ throughout the sequence of deletions.
To this end, we employ the result of Gawrychowski and Janczewski~\cite{GJ21} (see also~\cite{KS21}), who provide a dynamic algorithm maintaining a constant approximation of $\LIS_\prec(x)$ in the more general settings in which both insertions and deletions are allowed.

\begin{lemma}[{\cite[{Theorem 1 and subsequent discussion}]{GJ21}}]\label{lem:dynamic-lis}
There is a fully dynamic algorithm maintaining a $2$-approximation of $|\LIS_{\prec}(x)|$ with insertions and deletions from a sequence $x$ working in $\Otild(1)$ worst-case time per update.
Furthermore, if the returned approximation is $k$, then in $O(k)$ time, the algorithm can also provide an increasing subsequence of length $k$.
\end{lemma}

For a sequence $x$ and a set of symbols $\pi$, we define the following functions:
\begin{itemize} 
    \item $\Project(x, \pi)$ - the subsequence of $x$ consisting only of symbols in $\pi$.
    \item $\Exclude(x, \pi)$ - the subsequence of $x$ consisting only of symbols not in $\pi$.
\end{itemize}
A sequence $\pi=\pi_1,\pi_2,\dots,\pi_{|\pi|}$ is called \emph{repetition-free} sequence if for every $i\ne j \in[|\pi|]$ we have $\pi_i\ne\pi_j$.
When using $\Project$ and $\Exclude$, we often abuse notation by using a (repetition-free) sequence in place of a set of symbols. In such cases, we refer to the set implicitly defined by the sequence, namely $\{\pi_i \mid i \in [|\pi|]\}$.

We often interpret a repetition-free sequence $\pi$ as a total order over the symbols of the sequence.
Specifically, the total order $<_{\pi}$ over the symbols of $\pi$ is defined as $\sigma<_\pi \sigma'$ if and only if $\sigma=\pi_i$, $\sigma'=\pi_j$ and $i<j$, i.e., $\sigma$ occurs before $\sigma'$ in $\pi$.
We slightly abuse notation by writing $\LIS_{\pi}(x)$ (for a sequence $x$ over the symbols of $\pi$) to denote a longest increasing subsequence of $x$ with respect to the total order $<_{\pi}$ (instead of the more accurate and cumbersome $\LIS_{<_{\pi}}(x)$).
We further generalize the notation $\LIS_\pi(x)$ to be applicable to sequences $x$ over any alphabet.
To allow this, we define $\LIS_{\pi}(x) =\LIS_\pi(\Project(x,\pi))$ if $x$ contains symbols not in $\pi$.

\para{Erdős-Szekeres theorem.}
The following well-known lemma by Erd\H{o}s and Szekeres~\cite{ES35} will be useful for our algorithms.
\begin{lemma}[{Erd\H{o}s-Szekeres~\cite{ES35}}]\label{lem:ES}
Let $x$ be a repetition-free sequence of length $|x|$, and let $\prec$ be a total order on the elements of $x$.  
Let $i = |\LIS_\prec(x)|$ and $d = |\LDS_\prec(x)|$.  
Then, $i \cdot d \ge |x|$.
\end{lemma}

    


\section{Warm-up Algorithm}\label{sec:Warm-up}
In this section, we introduce a warm-up algorithm (see \cref{alg:warmup}).
The algorithm is composed of three parts. 
Each part computes a common subsequence of $x$ and $y$ which is a candidate for the output of the algorithm.
At the end, the output is the longest candidate.

At the first part of the algorithm, the algorithm computes $\BestMatch(x,y)$, which is a longest common subsequence of $x$ and $y$ among subsequences composed of a single symbol from $\Sigma$, as a candidate.

In the second part, the algorithm selects a repetition-free subsequence $\pi$ of $x$. 
Specifically, it considers the repetition-free subsequence $\pi=\RF(x)$, which consists of the first occurrence of each symbol in $x$.
We interpret $\pi$ as a total order over the symbols of $x$.
We note that any arbitrary repetition-free subsequence of $x$ containing all symbols of $x$ would be sufficient for our algorithm.
The algorithm computes an $\LIS$ of $\Project(y,\pi)$, with respect to $\pi$.
This $\LIS$ is also a common subsequence of $x$ and $y$, and is considered as a candidate for the output.

The third part of the algorithm iteratively finds $\pi'$, an approximate longest decreasing subsequence of $x$ with respect to $\pi$.
As in the second part, the algorithm interprets $\pi'$ as a total order over the symbols of $\pi'$, and computes an $\LIS$ of $\Project(y,\pi')$, with respect to $\pi'$, as a candidate.
At the end of the iteration, the algorithm removes all occurrences of symbols in $\pi'$ from $x$, and halts if $x$ becomes empty.
Finally, the algorithm returns the longest candidate found throughout all three parts.

In the pseudo-code below, we use the following notation. 
For two sequences $L$ and $L'$, the notation $L \stackrel{\max}{\gets} L'$ means that if $|L'| > |L|$, then $L$ is updated to $L'$; otherwise, $L$ remains unchanged.
In addition, we use the notation $\ALDS_\pi(x)$ to denote a 2-approximation of $\LDS_\pi(x)$ obtained by \cref{lem:dynamic-lis}.


\begin{algorithm}[H]
\caption{Approx-$\LCS(x,y)$}
\label{alg:warmup}

\tcp*[h]{Part 1: Best match}\;
$L \gets \BestMatch(x,y)$\label{line:bestmatch}\;

\tcp*[h]{Part 2: LIS with respect to $\pi$}\;
$\pi\gets \RF(x)$\label{line:getpifirst}\;
$L\stackrel\max\gets \LIS_{\pi}(y)$\label{line:lispifirst};

\tcp*[h]{Part 3: Greedy LDS peeling  of $x$}\;
\While(\label{line:while}){$x\ne\emptyset$}{
    $\pi'\gets \ALDS_{\pi}(x)$\label{line:ALDS}\;
    $L\stackrel\max\gets\LIS_{\pi'}(y)$\label{line:lispiprime}\;
    $x\gets \Exclude(x,\pi')$\label{line:exclude}\;
}

\Return $L$\;
\end{algorithm}

\para{Running Time}
We describe how to implement \cref{alg:warmup} in near-linear time.
Computing $\BestMatch(x,y)$ in \cref{line:bestmatch} can be implemented in $\Otild(n)$ time straight-forwardly.
Computing $\pi = \RF(x)$ in \cref{line:getpifirst} is also trivial to implement in $\Otild(n)$ time.
Finding $\LIS_{\pi}(y)=\LIS_{\pi}(\Project(y,\pi))$ is implemented by straightforwardly finding $\Project(y,\pi)$ in $\Otild(|y|+|\pi|)=\Otild(|y|+|x|)=\Otild(n)$ time, and then applying a near-linear time $\LIS$ algorithm to obtain  $\LIS_{\pi}(\Project(y,\pi))$.

We proceed to describe how to implement the while loop in \cref{line:while}.
An efficient computation of this loop boils down to developing an efficient data structure for finding all occurrences of a given symbol $\sigma$ in $x$ and in $y$.
Since $x$ undergoes deletions throughout the algorithm, this data structure needs to support symbol deletion as well.
We introduce the data structure in the following simple lemma, which is proved for the sake of completeness in \cref{app:occ-ds}.

\begin{restatable}{lemma}{OccDS}
\label{lem:occ-ds}
    There exists a data structure  maintaining a dynamic sequence $x$  supporting the following operations:
    \begin{itemize}
        \item $\Init(x)$ - initiallize the data structure; runs in $\Otild(|x|)$ time.
        \item $\Occurrences(\sigma) $ - returns all the indices in which the symbol $\sigma$ occur, i.e., $\{i\mid x_i=\sigma\}$; runs in   $\Otild(\#_\sigma(x))$ time.
        \item $\Delete(i)$ - deletes the $i$th element of $x$; runs in $\Otild(1)$ time.
    \end{itemize}
\end{restatable}

We now proceed to describe the implementation of the while loop.  
We begin by initializing the data structure $D_{\LIS}$ from \cref{lem:dynamic-lis} for $x$ with respect to the order $\pi$, as well as two instances of the data structure from \cref{lem:occ-ds}: $D_{\Occurrences}^x$ for $x$ and $D_{\Occurrences}^y$ for $y$.
At every iteration, the algorithm uses $D_{\LIS}$ to obtain $\pi'$ which is a $2$-approximation of $\LDS(x)$ in $O(|\pi'|)$ time.
The algorithm uses $D_\Occurrences^y$ to find all indices in $y$ in which symbols of $\pi'$ occur in $O(|\sum_{\sigma \in \pi'}\#_\sigma(y)|)= \Otild(|\Project(y,\pi')|)$ time.
The algorithm iterates these indices and concatenates the corresponding symbols from $y$ to obtain $\Project(y,\pi')$ in $\Otild(|\Project(y,\pi')|)$ time.
Then, the algorithm applies a near linear time $\LIS$ algorithm to compute $\LIS_{\pi'}(\Project(y,\pi'))$.
Next, the algorithm finds all indices in $x$ containing a symbol of $\pi'$ in $\Otild(|\sum_{\sigma \in \pi'}\#_\sigma(x)|)$ time.
The algorithm then updates both $D_{\LIS}$ and $D_{\Occurrences}^x$ to delete all indices in $x$ that contain occurrences of symbols from $\pi'$.
This implements $x\gets \Exclude(x,\pi')$. 
The symbols of $\pi'$ in each iteration are disjoint from the symbols of $\pi'$ in every  previous iteration.
Therefore, the total running time for implementing \cref{line:ALDS,line:lispiprime,line:exclude} is bounded by $\Otild(\sum_{\sigma \in \Sigma}\#_\sigma(x) +\#_\sigma(y)) = \Otild(|x|+|y|)=\Otild(n)$.

\para{Correctness.}
We show that each candidate for $L$ is a common subsequence of $x$ and $y$.
For $L=\BestMatch(x,y)$ this is trivial.
The rest of the candidates are created by taking $\eta$, a repetition-free subsequence of $x$,  and finding an increasing subsequence of $y$ with respect to $\eta$.
Clearly, the result is a subsequence of $y$, and it is also a subsequence of $x$ as a subsequence of $\eta$.

\para{Approximation ratio.}
Clearly, if $|\LCS(x,y)| = 0$ the algorithm would not find any common subsequence, which is a satisfactory answer.
If $|\LCS(x,y)| \ge 1$, in particular there is a common symbol occurring both in $x$ and in $y$.
It immediately follows that $\BestMatch(x,y)$ returns a common subsequence of length at least $1$, which is an $n^{4/5}$-approximation if $|\LCS(x,y)| \le n^{4/5}$.
Let us assume from now on that $|\LCS(x,y)| > n^{4/5}$. In particular, let $0< t\le 1/5$ be the number such that $|\LCS(x,y)|=n^{4/5+t}$.

If $|\BestMatch(x,y)|\ge n^{t}$, we have found a candidate which is an $n^{4/5}$-approximation for $\LCS(x,y)$.
Let us assume from now on that $|\BestMatch(x,y)| < n^{t}$.
This means that no symbol occurs in both $x$ and $y$ more than $n^{t}$ times.
Let $L$ be some longest common subsequence of $x$ and $y$.
In particular, no letter occurs in $L$ more than $n^{t}$ times.
Consider the subsequence $\RF(L)$ of $L$ in which an arbitrary occurrence of each symbol is taken and everything else is deleted.
Since every symbol occurs in $L$ less than $n^{t}$ times and $|L| = n^{4/5+t}$, we have that $|\RF(L)| > n^{4/5}$.

Assume that $|\LIS_\pi(\RF(L))|\ge n^{t}$.
Since $\RF(L)$ is a subsequence of $y$, we have in particular that $|\LIS_\pi(y)|\ge |\RF(L)|\ge  n^{t}$. 
In this case, \cref{line:lispifirst} returns a candidate of length at least $n^{t}$, which is an $n^{4/5}$-approximation for every possible length of $\LCS(x,y)$.

Assume from now on that $|\LIS_\pi(\RF(L))|<n^{t}$.
By Erdős-Szekeres theorem (\cref{lem:ES}), it holds that $|\LDS_\pi(\RF(L))|\ge |\RF(L)|/n^{t} \ge n^{4/5-t}$.

Before proceeding to analyze the while loop in \cref{line:while}, we provide an intuitive explanation for why it should work.  
Let $D = \LDS_{\pi}(\RF(L))$ be the longest decreasing subsequence of the repetition-free reduction of the $\LCS$ according to $\pi$.  
As a common subsequence, $D$ appears in both $x$ and $y$.

At each iteration of the while loop, the algorithm finds an approximate longest decreasing subsequence $\pi'$ of $x$ with respect to $\pi$, and removes from $x$ all symbols participating in $\pi'$.  
Initially, $D$ is a valid candidate for the longest decreasing subsequence of $x$ with respect to $\pi$.

Notice that the symbols of $\pi'$ appear in the same order in $\pi'$ and in $D$ (i.e., decreasing in $\pi$), which is a subsequence of $y$.  
We therefore say that an iteration in which $\pi'$ contains many symbols of $D$ is \emph{good} - such an iteration would yield a large candidate $\LIS_{\pi'}(y)$ for the $\LCS$.

If an iteration is not \emph{good}, it removes only a small fraction of $D$, so in the next iteration, $D$ remains fairly large, and the same argument can be applied again.
By repeatedly applying this reasoning, we conclude that $x$ is reduced relatively quickly - $D$ is barely affected by a sequence of \emph{bad} iterations and therefore remains a candidate for the next $\pi'$.
By carefully selecting a threshold that distinguishes \emph{good} from \emph{bad} iterations, we show that if all iterations are \emph{bad}, then $x$ is completely deleted after a certain number of iterations.  
On the other hand, under this assumption, some symbols of $D$ must still remain in $x$ after this many iterations - a contradiction.  
Hence, we conclude that at least one iteration must be \emph{good}.

We now proceed to formally bound the approximation ratio achieved by the algorithm.
Consider the $i$th iteration of the while loop in \cref{line:while}.
Denote by $\pi^i$ the value of $\pi'$ in this iteration, let $\ell_i$ be the length of $\LIS_{\pi'}(y)$ in \cref{line:lispiprime} in this iteration, and let $x^i$ be the value of $x$ at the beginning of the iteration ($x^0=x$).
We say that the $i$th iteration is \emph{bad}, if $\ell_i<n^{t}/4$.
Clearly, if there is some iteration that is not bad, the algorithm finds a candidate that is a $4n^{4/5}$-approximation of $\LCS(x,y)$.
In the following lemma we prove that there is an iteration that is not bad.
This concludes the proof that \cref{alg:warmup} is a $4n^{4/5}$-approximation algorithm for the $\LCS$ problem.

\begin{claim}\label{lem:notbadexists}
    There is at least one iteration that is not bad.
\end{claim}
\begin{claimproof}
Assume by contradiction that all the iterations performed by the algorithm are bad.
Let $D=\LDS_\pi(\RF(L))$, and for every iteration $i$ let 
$D_i=\Exclude(D,\bigcup_{j=0}^{i-1}\pi^j)=\Exclude(D_{i-1},\pi^{i-1})$, and let 
$x^i=\Exclude(x,\bigcup_{j=0}^{i-1}\pi^j)=\Exclude(x^{i-1},\pi^{i-1})$.
Notice that $D_i$ is a subsequence of $x^i$ which is decreasing with respect to $\pi$.
Since $\pi^i$ is a 2-approximation of $\LDS_\pi(x^i)$, it holds that $|\pi^i|\ge |\LDS_\pi(x^i)|/2\ge |D_i|/2$.

Let $\tau^i=\Project(D_i,\pi^i)$.
Since $\tau^i$ is a subsequence of $D$, $\tau^i$ is a decreasing  subsequence  of $y$ with respect to $\pi$.
Since $\pi^i$ is also decreasing with  repsect to $\pi$, then $\tau^i$ is increasing with respect to $\pi^i$.
It follows that $\ell_i=|\LIS_{\pi^i}(y)|\ge |\tau^i|$.
Since the $i$th iteration is bad, $|\tau^i|\le \ell_i\le n^{t}/4$.
It follows that $|D_{i+1}|=|\Exclude(D_i,\pi^i)|=|D_i|-|\Project(D_i,\pi^i)|= |D_i|-|\tau^i|\ge |D_i|-n^{t}/4$.
By applying the above inequality inductively, we obtain $|D_i|\ge|D|-i\cdot n^{t}/4>n^{4/5-t}-i\cdot n^{t}/4$.

On the other hand, we have that
\[|x^{i+1}|=|\Exclude(x^{i},\pi^{i})|\le |x^{i}|-|\pi^{i}|\le |x^{i}|-|D_{i}|/2,\] where the first inequality holds since each symbol of $\pi^i$ occurs in $x^i$ at least once (as $\pi^i$ is a subsequence of $x^i$).
By applying the above inequality inductively and observing that $|D_1|,|D_2|,\dots, |D_i|$ is a monotone decreasing sequence, we obtain $|x^i|\le |x|-i\cdot |D_i|/2$.

Let $z=2n^{1/5+t}$.
On one hand, the $z$th iteration of the algorithm occurs, since $|D_z|\ge |D|-2n^{1/5+t}\cdot n^{t}/4\ge n^{4/5-t}-n^{1/5+2t}/2\ge n^{4/5-t}/2$ where the last inequality holds since $t\le 1/5$.
In particular, $x^z$ is not empty.
On the other hand, the length of $x^z$ in this iteration is
$|x^z|\le |x|-2n^{1/5+t}\cdot |D_z|/2\le n-2n^{1/5+t}\cdot n^{4/5-t}/2=n-n=0$, in contradiction.
\end{claimproof}

\section{Better Algorithm - Proof of \cref{thm:main}}\label{sec:main}
In this section, we present an improved approximation algorithm that achieves an $O(n^{3/4} \log n)$-approximation to $\LCS(x, y)$.

For an interval of numbers $[a..b]$ denote $\Sigma_{[a..b]}(x)=\{\sigma\in \Sigma\mid \#_\sigma(x)\in [a..b]\}$.
The following pseudocode describes the algorithm.

\begin{algorithm}[H]
\caption{Better-Approx-$\LCS(x,y)$}
\label{alg:better}
\setcounter{AlgoLine}{0}

\ForEach(\label{line:foreach}){$f\in\{2^i\mid i\in[0..\floor{\log n}]\}$}{
    $x'\gets \Project(x,\Sigma_{[f..n]}(x))$;
    
    
    $L\stackrel\max\gets$Approx-$\LCS(x',y)$; \tcp*{Call to \cref{alg:warmup}}
}
\Return $L$\;
\end{algorithm}

\para{Running Time.}
The while loop of the algorithm runs $O(\log n)$ times.
In each iteration, filtering out the least frequent symbols can be straight-forwardly implemented in $\Otild(n)$ time.
Then, running algorithm Approx-$\LCS(x',y)$ takes $\Otild(n)$ time.
Thus, the total running time is $\Otild(n)$.

\para{Correctness.}
The algorithm returns a common subsequence obtained by Approx-$\LCS(x',y)$ for some subsequence $x'$ of $x$ due to the correctness of \cref{alg:warmup}.
Since $x'$ is a subsequence of $x$, the output is a common subsequence of $x$ and $y$.

\para{Approximation Ratio.}
Let $L$ be some longest common subsequence of $x$ and $y$. 
We observe that for some power $2^i$ of $2$, a significant fraction of $L$ consists of symbols that appear roughly $2^i$ times in $L$.


\begin{lemma}\label{lem:goodf}
    For some $f \in \{2^i \mid i \in [0..\floor{\log n}]\}$, we have
    \[
    |\Project(L, \Sigma_{[f..2f)}(L))| \ge \frac{|L|}{2\log n} .
    \]
\end{lemma}
\begin{proof}
    The claim follows directly from the pigeonhole principle. 
    Formally,
    \begin{align*}
    |L| = \sum_{\sigma \in \Sigma} \#_\sigma(L)
    &= \sum_{f \in \{2^i \mid i \in [0..\floor{\log n}]\}} \sum_{\sigma \in \Sigma_{[f..2f)(L)}} \#_\sigma(L)
    \\&= \sum_{f \in \{2^i \mid i \in [0..\floor{\log n}]\}} |\Project(L, \Sigma_{[f..2f)}(L))|.
    \end{align*}
    By the pigeonhole principle, at least one of the summands must be at least $|L|$ divided by the number of summands, which is at most $1+\log n\le 2\log n$. 
    Thus, the claim follows.
\end{proof}

Let $f$ be  some $f \in \{2^i \mid i \in [0..\floor{\log n}]\}$, such that
    $|\Project(L, \Sigma_{[f..2f)}(L))| \ge  \frac{|L|}{2\log n}$ (the existence of $f$ follows from \cref{lem:goodf}).
    Denote $L'=\Project(L, \Sigma_{[f..2f)}(L))$. 
    We proceed to analyze the approximation ratio achieved by running Approx-$\LCS(x',y)$ for
    $x'=\Project(x,\Sigma_{[f,n]}(x))$.
    In particular we will show that Approx-$\LCS(x',y)$ returns an $O(n^{3/4}\log n)$-approximation for $\LCS(x,y)$.
    A useful property of $x'$, is that for any $\sigma\in\Sigma$ that occurs in $x$ we have $\#_\sigma(x)\ge f$.

Clearly, if $|\LCS(x,y)| = 0$ the algorithm would not find any common subsequence, which is a satisfactory answer.
If $|\LCS(x,y)| \ge 1$, in particular there is a common symbol occurring both in $x'$ and in $y$.
It immediately follows that $\BestMatch(x',y)$ returns a common subsequence of length at least $1$, which is an $n^{3/4}$-approximation if $|\LCS(x,y)| \le n^{3/4}$.
Let us assume from now on that $
|L|=|\LCS(x,y)| > n^{3/4}$. 
In particular, let $0< t\le 1/4$ be the number such that $|L|=n^{3/4+t}$.


Assume that $f\ge n^{1/4}$.
Then, in particular, there is some symbol $\sigma\in \Sigma_{[f..2f)}(L)\subseteq\Sigma_{[f..n)}(x)$ such that $\#_\sigma(L')\ge f\ge n^{1/4}$, and $|\BestMatch(x',y)|\ge \Match(x',y,\sigma)\ge \#_\sigma(L')\ge n^{1/4}$.
It follows that  $\BestMatch(x',y)$ is an $n^{3/4}$-approximation of $\LCS(x,y)$.
We therefore assume from now on that $f< n^{1/4}$.

Recall that $L'=\Project(L, \Sigma_{[f..2f)}(L))$.
Consider the subsequence $\RF(L')$ of $L'$ in which an arbitrary occurrence of each symbol is taken and everything else is deleted.
Since every symbol occurs in $L'$ less than $2f$ times, we have that $|\RF(L')| > |L'|/2f\ge \frac{|\LCS(x,y)|}{2\log n}/2f = \frac{n^{3/4+t}}{4f\log n}$.

Assume that $|\LIS_\pi(\RF(L'))|\ge n^{t}/\log n$.
Since $\LIS_{\pi}(\RF(L'))$ is a subsequence of $y$, we have in particular that  $|\LIS_\pi(y)|\ge  n^{t}/\log n$. 
In this case, \cref{line:lispifirst} returns a candidate of length at least $n^{t}/ \log n$, which is an $n^{3/4} \log n$-approximation for $\LCS(x,y)$.

Assume from now on that $|\LIS_\pi(\RF(L'))|<n^{t}/\log n$.
By Erdős-Szekeres theorem (\cref{lem:ES}), it holds that $|\LDS_\pi(\RF(L'))|\ge |\RF(L')|/(n^{t}/\log n) \ge \frac{n^{3/4}}{4f}$.
Consider the $i$th iteration of the while loop in \cref{line:while}.
Denote $\pi^i$ as the value of $\pi'$ in this iteration, let $\ell_i$ be the length of $\LIS_{\pi'}(y)$ in \cref{line:lispiprime} in this iteration, and let $x^i$ the value of $x'$ at the beginning of the iteration ($x^0=x'$).
We say that the $i$th iteration is \emph{bad}, if $\ell_i<\frac{n^{t}}{200}$.
Clearly, if there is some iteration that is not bad, the algorithm finds a candidate that is an $O(n^{3/4})$-approximation of $\LCS(x,y)$.
In the following lemma we prove that there is an iteration that is not bad.
This concludes the proof that \cref{alg:better} is an $O(n^{3/4}\log n)$-approximation algorithm for the $\LCS$ problem, proving \cref{thm:main}.

\begin{claim}\label{lem:notbadexists_better}
    There is at least one iteration that is not bad.
\end{claim}
\begin{claimproof}
Assume by contradiction that all the iterations performed by the algorithm are bad.
Let $D=\LDS_\pi(\RF(L'))$, let 
$D_i=\Exclude(D,\bigcup_{j=0}^{i-1}\pi^j)=\Exclude(D_{i-1},\pi^{i-1})$, and let 
$x^i=\Exclude(x',\bigcup_{j=0}^{i-1}\pi^j)=\Exclude(x^{i-1},\pi^{i-1})$.
Notice that $D_i$ is a subsequence of $x^i$ which is decreasing with respect to $\pi$.
Since $\pi^i$ is a 2-approximation of $\LDS_\pi(x^i)$, it holds that $|\pi^i|\ge |\LDS_\pi(x^i)|/2\ge |D_i|/2$.

Let $\tau^i=\Project(D_i,\pi^i)$.
Since $\tau^i$ is a subsequence of $D$, $\tau^i$ is a decreasing  subsequence  of $y$ with respect to $\pi$.
Since $\pi^i$ is also decreasing with  respect to $\pi$, then $\tau^i$ is increasing with respect to $\pi^i$.
It follows that $\ell_i=|\LIS_{\pi^i}(y)|\ge |\tau^i|$.
Since the $i$th iteration is bad, $|\tau^i|\le \ell_i\le \frac{n^{t}}{200}$.
It follows that $|D_{i+1}|=|\Exclude(D_i,\pi^i)|=|D_i|-|\Project(D_i,\pi^i)|= |D_i|-|\tau^i|\ge |D_i|-\frac{n^{t}}{200}$.
By applying the above inequality inductively, we obtain $|D_i|\ge|D|-i\cdot \frac{n^{t}}{200}>\frac{n^{3/4}}{4f}-i\cdot \frac{n^{t}}{200}$.

On the other hand, we have that
\[|x^{i+1}|=|\Exclude(x^{i},\pi^{i})|\le |x^{i}|-|\pi^{i}|\cdot f\le |x^{i}|-|D_{i}|\cdot f/2,\] where the first inequality holds since each symbol of $\pi^i$ occurs in $x^i$ at least $f$ times (as $\pi^i$ is a subsequence of $x^i$).
By applying the above inequality inductively and observing that $|D_1|,|D_2|,\dots, |D_i|$ is a monotone decreasing sequence, we obtain $|x^i|\le |x'|-i\cdot |D_i|\cdot f/2$.

Let $z=25n^{1/4}$.
On one hand, the $z$th iteration of the algorithm occurs, since $|D_z|\ge |D|-25n^{1/4}\cdot \frac{n^{t}}{200}\ge \frac{n^{3/4}}{4f}-\frac{n^{1/4+t}}{8}\ge \frac{n^{3/4}}{8f\log n}$ where the last inequality holds since $t\le 1/4$ and $f\le n^{1/4}$.
In particular, $x^z$ is not empty.
On the other hand, the length of $x^z$ in this iteration is
$|x^z|\le |x|-25n^{1/4}\cdot\frac{n^{3/4}}{8f}\cdot f/2\le n-25/16\cdot n\le 0$, in contradiction.
\end{claimproof}


    


\bibliography{ref.bib}

\appendix

\section{Symbols Occurrences Data Structure (Proof of \cref{lem:occ-ds})}\label{app:occ-ds}
In this section, we prove \cref{lem:occ-ds}; restated below.

\OccDS*
\begin{proof}
    
We implement the data structure as follows.
We initialize a self-balancing search tree (e.g. AVL tree or Red-Black tree) $T_I$ with $|x|$ elements, where initially the $i$th element of $T_I$ corresponds to the $i$th element $x_i$ of $x$.
Every node of $T_I$ also stores as auxiliary information the size of the sub-tree below it in $T_I$.
This information can be straightforwardly maintained in $O(\log n)$ time when $T_I$ undergoes a deletion, and it can be used to decide the current rank of a node among all remaining nodes in $O(\log n)$ time.

For each symbol $\sigma \in \Sigma$, the algorithm initializes a balanced search tree $T_{\sigma}$ storing all indices in $x$ in which $\sigma$ occurs.
The indices in $x$ are not stored explicitly as integers, but as pointers to the elements of $T_I$ - so when an index $i$ is deleted from $T_I$, all remaining indices larger than $i$ are implicitly shifted accordingly.
Clearly, given a letter $\sigma$ we can find all indices in $x$  in which $\sigma$ occurs in $\Otild(\#_\sigma(x))$ time  using $T_I$ and $T_{\sigma}$.
Deletion of the $i$th element is implemented by removing the $i$th element of $T_i$ and the corresponding element in $T_{x_i}$.
\end{proof}

\end{document}